\documentclass[12pt]{amsart}
\usepackage{amssymb}
\usepackage{color}
\usepackage{amsmath,epic,curves,amscd}
\usepackage[english]{babel}
\usepackage{graphicx}
\usepackage{comment}
\usepackage{appendix}
\usepackage{mathdots}
\usepackage[all]{xy}
\usepackage{mathtools}
\usepackage{lscape}
\usepackage{stmaryrd}
\usepackage{mathabx}
\newtheorem{claim}{}[section]
\newtheorem{theorem}[claim]{Theorem}

\newtheorem{proposition}[claim]{Proposition}

\theoremstyle{remark}

\renewenvironment{proof}{\noindent{\it Proof. \hskip0pt}}
                      {$\square$\par\medskip}
\allowdisplaybreaks \textwidth 15.5 true cm \textheight 23.9 true cm
\begin{document}
\baselineskip 6.0 truemm
\parindent 1.5 true pc

\newcommand\lan{\langle}
\newcommand\ran{\rangle}
\newcommand\tr{{\text{\rm Tr}}\,}
\newcommand\ot{\otimes}
\newcommand\ttt{{\text{\sf t}}}
\newcommand\rank{\ {\text{\rm rank of}}\ }
\newcommand\choi{{\rm C}}
\newcommand\dual{\star}
\newcommand\flip{\star}
\newcommand\cp{{\mathcal C}{\mathcal P}}
\newcommand\ccp{{\mathcal C}{\mathcal C}{\mathcal P}}
\newcommand\pos{{\mathcal P}}
\newcommand\tcone{T}
\newcommand\mcone{K}
\newcommand\superpos{{\mathbb S\mathbb P}}
\newcommand\blockpos{{\mathcal B\mathcal P}}
\newcommand\jc{{\text{\rm JC}}}
\newcommand\dec{{\mathcal D}{\mathcal E}{\mathcal C}}
\newcommand\ppt{{\mathcal P}{\mathcal P}{\mathcal T}}
\newcommand\xxxx{\bigskip\par ================================}
\newcommand\join{\vee}
\newcommand\meet{\wedge}
\newcommand\ad{{\text{\rm Ad}}\,}
\newcommand\HS{{\text{\rm HS}}}
\newcommand\sr{{\text{\rm SR}}\,}
\newcommand\e{\varepsilon}

\title{Supporting hyperplanes for Schmidt numbers and Schmidt number witnesses}

\author{Kyung Hoon Han and Seung-Hyeok Kye}
\address{Kyung Hoon Han, Department of Data Science, The University of Suwon, Gyeonggi-do 445-743, Korea}
\email{kyunghoon.han at gmail.com}
\address{Seung-Hyeok Kye, Department of Mathematics and Institute of Mathematics, Seoul National University, Seoul 151-742, Korea}
\email{kye at snu.ac.kr}

\keywords{supporting hyperplanes, $k$-blockpositive matrices, Schmidt number witnesses, Werner states, isotropic states}
\subjclass{15A30, 81P15, 46L05, 46L07}

\begin{abstract}
We consider the compact convex set of all bi-partite states of Schmidt number less than or equal to $k$,
together with that of $k$-blockpositive matrices of trace one, which play the roles of Schmidt number witnesses.
In this note, we look for \mbox{hyperplanes} which support those convex sets and are perpendicular
to a one parameter family through the maximally mixed state. We show that this is equivalent to
determining the intervals for the dual objects on the one parameter family.
We illustrate our results for the one parameter families including Werner states and isotropic states.
Through the discussion, we give a simple decomposition of the separable Werner state into the sum of product states.
\end{abstract}
\maketitle

\section{Introduction}

Entanglement \cite{Werner-1989} is now indispensable in the current quantum information theory,
and the notion of Schmidt numbers \cite{terhal-schmidt} of bi-partite states is an important tool to measure the degree of entanglement.
The class of $k$-blockpositive matrices \cite{{jam_72},{ssz}} plays the role of witnesses \cite{{lkhc},{SBL_2001}}
to determine entanglement and Schmidt numbers
through the bilinear pairing between Hermitian matrices, and it is important to understand the facial structures of
the compact convex set $\blockpos_k$ of all $k$-blockpositive matrices of trace one. For example,
an entanglement witness $W\in\blockpos_1$ is optimal (respectively has the spanning property)
if and only if the smallest face (respectively smallest
exposed face) containing $W$ has no positive matrix \cite{{ha-kye-optimal},{kye_ritsu}}.
In this regard, we consider supporting hyperplanes
to the convex set $\blockpos_k$, as well as Schmidt numbers themselves.
We recall that $k$-blockpositive matrices are just Choi matrices of $k$-positive linear maps,
and the facial structures for the convex cone of all $k$-positive maps have been studied in
\cite{{kye-canad},{kye-korean},{kye-cambridge}}. See \cite{{kye_ritsu},{kye_lec_note}} for survey articles
on the related topics.

We recall that a closed half-space in an affine space is called {\sl a supporting half-space} for a convex set $C$
when it contains $C$ entirely and shares at least one boundary point with $C$.
The boundary of a supporting half-space for $C$ is called a {\sl supporting hyperplane} to $C$.
See {\sc Figure 1}.
Supporting hyperplanes to a convex set give rise to exposed faces of the convex set by taking intersections with it,
and every exposed face arises in this way.

\begin{figure}
\begin{center}
$$
\includegraphics[scale=0.7]{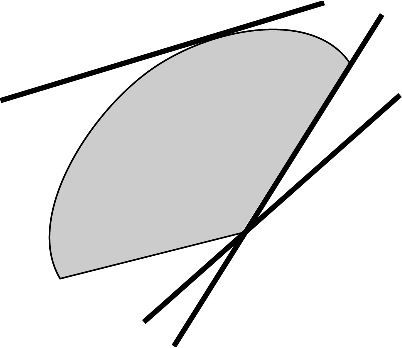}
$$
\end{center}
\caption{Thick lines represent supporting hyperplanes to the convex set.}
\end{figure}

In this note, we look for supporting hyperplanes to $\blockpos_k$ which are perpendicular to the one parameter family
\begin{equation}\label{one-para}
X_\lambda=(1-\lambda)\varrho_*+\lambda\varrho\in M_m\ot M_n,\qquad -\infty<\lambda<+\infty
\end{equation}
of Hermitian matrices, where $\varrho_*=\frac 1{mn}I_{mn}$ denotes the  maximally mixed state, and
$\varrho\in M_m\ot M_n$ is a bi-partite state different from $\varrho_*$.
This will be, in fact, equivalent to know to what extent $X_\lambda$ is a state of Schmidt number
at most $k$. We use the notation $X^\varrho_\lambda$
when we emphasize the role of $\varrho$.
Note that every Hermitian matrix of trace one
is written by $X^\varrho_\lambda$ for a state $\varrho$ and a real number $\lambda$.

For our purpose, we will begin with the interval $[\beta^-_k,\beta^+_k]$
satisfying the condition that $X_\lambda \in\blockpos_k$ if and only if $\beta^-_k\le\lambda\le\beta^+_k$,
as it was considered in \cite{han_kye_beta}.
We denote by $\mathcal H^0_\nu$  the hyperplane  through $X_\nu$ which is perpendicular to the one parameter family $\{X_\lambda\}$
and determine two numbers $\tilde\beta^+_k$ and $\tilde\beta^-_k$ such that $\mathcal H^0_{\tilde\beta^+_k}$ and $\mathcal H^0_{\tilde\beta^-_k}$
are supporting hyperplanes to $\blockpos_k$. The inequalities $\tilde\beta^-_k\le \beta^-_k$
and $\beta^+_k\le \tilde\beta^+_k$ are clear.
We also consider the interval  $[\sigma^-_k,\sigma^+_k]$ satisfying
$$
X_\lambda \in\mathcal S_k\ \Longleftrightarrow \sigma^-_k\le\lambda\le\sigma^+_k,
$$
for the compact convex set $\mathcal S_k$ of all states of Schmidt number less than or equal to $k$.
Our main result tells us that
$$
\lan X_{\tilde\beta^+_k}|X_{\sigma^-_k}\ran=0\qquad {\text{\rm and}}\qquad \lan X_{\tilde\beta^-_k}|X_{\sigma^+_k}\ran=0
$$
hold, and so we see that determining the intervals $[\tilde\beta^-_k,\tilde\beta^+_k]$ and
$[\sigma^-_k,\sigma^+_k]$ are equivalent problems.
In fact, we will show this in Theorem \ref{main} for general compact convex sets and their dual convex sets
with a minor restriction.
We recall that $X_\lambda\in\mathcal S_k$ if and only if $\lan X | X_\lambda \ran \ge 0$ for every $X\in\blockpos_k$.
Therefore, it is natural to consider the affine function
\begin{equation}\label{funct}
f_\lambda:X\mapsto \lan X | X_\lambda \ran,\qquad X\in\mathcal H,
\end{equation}
on the affine space $\mathcal H$ of all Hermitian matrices in $M_m\ot M_n$ of trace one. We are motivated by the
simple observation for the function $f_\lambda$:
All the level sets of $f_\lambda$ are perpendicular to the line represented by the one parameter family $\{X_\lambda\}$.
This is why we are looking for supporting hyperplanes to $\blockpos_k$ which is perpendicular
to the one parameter family $\{X_\lambda\}$, in order to find the numbers $\sigma^+_k$ and $\sigma^-_k$
which determine Schmidt numbers of states.

We illustrate our result in the case when $\varrho=|\xi\ran\lan\xi|$ is a pure state. By the Schmidt decomposition, we may
assume that $|\xi\ran=\sum_{i=0}^{n-1}p_i |ii\ran \in\mathbb C^n\ot\mathbb C^n$ without loss of generality,
with $\sum_{i=0}^{n-1}p_i^2=1$ and $p_0\ge p_1 \ge\dots\ge  p_{n-1} \ge 0$.
When Schmidt coefficients $\{p_i\}$ are evenly distributed,
$\varrho$ is the maximally entangled state and
$X_\lambda$'s give rise to the isotropic states \cite{terhal-schmidt}.
We also have the Werner states \cite{Werner-1989} by taking the partial transpose.
By the results in \cite{han_kye_beta}, we first find the number $\beta^-_1$ and optimize the witnesses $X_{\beta^-_1}$
by subtracting a diagonal matrix with nonnegative entries.
By decomposing the resulting witness into the sum of extreme witnesses, we have natural candidates which may give rise to
the number $\tilde\beta^-_1$. We take the minimum among them to get numbers $\tilde\beta^-_1$  and $\sigma^+_1$.

After we develop general principles to find supporting hyperplanes perpendicular to the one parameter family
in the next section, we consider the
variations of the isotropic states and Werner states mentioned above in Section 3.
Through the discussion, we give a simple decomposition of the separable Werner state into the sum of product states.
Such decompositions have been considered by several authors.
See \cite{{wootters},{azumi_ban},{UKB},{Gr_Ap},{li_qiao},{PPPR},{YLQ}}.

Parts of this note were presented by the second author at \lq\lq Mathematical Structures
in Quantum Mechanics\rq\rq\ which was held at Gdansk, Poland, in March 2025.
He is grateful to the audience for stimulating discussions and organizers Adam Rutkowski and Marcin Marciniak for travel support.

\section{Supporting hyperplanes}

We work on the affine space $\mathcal H$ in $M_m\ot M_n$ consisting of all Hermitian matrices with trace one.
For a given state $\varrho \ne \varrho_*$, we consider the one parameter family $\{X_\lambda\}$ defined by (\ref{one-para}).
For a nonzero real number $\lambda$, we take the affine function (\ref{funct})
on $\mathcal H$. Then we have
\begin{equation}\label{id}
\begin{aligned}
f_\lambda(X_\nu)=\lan X_\nu | X_\lambda \ran
&=\lan(1-\nu)\varrho_*+\nu\varrho | (1-\lambda)\varrho_*+\lambda\varrho \ran\\
&=\tfrac 1{mn}[(1-\nu)(1-\lambda)+(1-\nu)\lambda+\nu(1-\lambda)]+\nu\lambda\|\varrho\|^2_\HS\\
&=(\|\varrho\|_\HS^2-\textstyle\frac 1{mn})\nu\lambda+\textstyle\frac 1{mn}.
\end{aligned}
\end{equation}
Since $\|\varrho\|_\HS^2-\textstyle\frac 1{mn}>0$, we see that $\nu\mapsto f_\lambda(X_\nu)$ is
a strictly increasing affine function which sends $0$ to $\frac 1{mn}$ whenever $\lambda>0$.
It is decreasing when $\lambda<0$.

\begin{proposition}
The level set $f^{-1}_\lambda(\alpha)$ for a nonzero real number $\lambda$ is perpendicular to the one parameter family $\{X_\lambda\}$.
\end{proposition}

\begin{proof}
Take $\nu$ such that $f_\lambda(X_\nu)=\alpha$. Then we have
$X\in f^{-1}_\lambda(\alpha)$ if and only if
$\lan X-X_\nu|X_\lambda\ran=0$ if and only if
$\lan X-X_\nu|X_\lambda-\varrho_*\ran=0$,
which means that $X-X_\nu$ is perpendicular to the one parameter family  since $X_\lambda-\varrho_*\neq 0$.
\end{proof}

For a closed convex set $C$ in $\mathcal H$, we define the dual $C^\circ$ by
$$
C^\circ=\{X\in\mathcal H: \lan X|Y\ran\ge 0\ {\text{\rm for every}}\ Y\in C\},
$$
which is also a closed convex set.
Recall that $\blockpos_k$ and $\mathcal S_k$ are dual to each other, that is, we have $\blockpos_k^\circ=\mathcal S_k$
and $\mathcal S_k^\circ=\blockpos_k$.

\begin{proposition}
Suppose that $C$ is a closed convex subset of the affine space $\mathcal H$.
Then we have the following;
\begin{enumerate}
\item[(i)] $\varrho_*$ is an interior point in $C$ if and only if $C^\circ$ is compact,
\item[(ii)] $C$ is compact if and only if $\varrho_*$ is an interior point in $C^\circ$,
\item[(iii)] if $C$ is compact and $\varrho_*$ is an interior point of $C$, then we have $C = C^{\circ\circ}$.
\end{enumerate}
\end{proposition}

\begin{proof}
We first prove the \lq only if\rq\ parts of (i). Suppose that $C$ contains a ball $B(\varrho_*,\varepsilon)$ and $W \in C^\circ$.
By
$$
0 \le \lan W | \left(\varrho_* \pm {\tfrac \varepsilon  2} |i\ran\lan j| \pm {\tfrac \varepsilon 2} |j\ran\lan i|\right)  \ran
= {\tfrac 1 {mn}} \pm \varepsilon {\rm Re} W_{ij},\qquad i\neq j,
$$
we see that the off-diagonal entries of $W$ satisfy
$$
- {\tfrac 1 {\varepsilon mn}} \le {\rm Re} W_{ij} \le {\tfrac 1 {\varepsilon mn}},\qquad i\neq j.
$$
Similarly,
$$
0 \le \lan W ~|~ \left( \varrho_* \pm {\tfrac \varepsilon 2} {\rm i} |i\ran\lan j| \mp {\tfrac\varepsilon  2} {\rm i} |j\ran\lan i| \right) \ran
= {\tfrac 1 {mn}} \pm \varepsilon {\rm Im} W_{ij}
$$
implies
$$
- {\tfrac 1 {\varepsilon mn}} \le {\rm Im} W_{ij} \le {\tfrac 1 {\varepsilon mn}},\qquad i\neq j.
$$

On the other hand,
$$
0 \le \lan W | \left(\varrho_* \pm {\tfrac \varepsilon  2} |i\ran\lan i| \mp {\tfrac \varepsilon  2} |j\ran\lan j|\right)  \ran
= {\tfrac 1  {mn}} \pm {\tfrac\varepsilon  2} W_{ii} \mp {\tfrac \varepsilon  2} W_{jj}
$$
implies
$- {\tfrac 2 {\varepsilon mn}} + W_{jj} \le W_{ii} \le {\tfrac 2 {\varepsilon mn}} + W_{jj}$.
Summing over all $1 \le j \le mn$, we see that diagonal entries satisfy
$$
- {\tfrac 2 {\varepsilon mn}} + {\tfrac 1 {mn}} \le W_{ii} \le {\tfrac 2 {\varepsilon mn}} + {\tfrac 1 {mn}},
$$
and this proves the \lq only if\rq\ part of (i).

For the \lq if\rq\ part of (i), we suppose that $\varrho_*$ is not an interior point.
If $\varrho_*$ is a boundary point then we take a supporting hyperplane through $\varrho_*$;
if $\varrho_*\notin C$ then we take a hyperplane separating $C$ and $\varrho_*$, and take a parallel
hyperplane through $\varrho_*$. In any cases, there exists a hyperplane $\mathcal H^0$ through $\varrho_*$
such that $C$ is contained in the closed half-space $\mathcal H^0\sqcup\mathcal H^+$ of $\mathcal H$. We take $\varrho\in\mathcal H^+$
so that $\varrho-\varrho_*$ is perpendicular to the hyperplane $\mathcal H^0$,
and consider the one parameter family $X_\lambda=(1-\lambda)\varrho_*+\lambda\varrho$. For any $X\in \mathcal H^0\sqcup\mathcal H^+$,
we take $\nu$ such that $0=\lan X-X_\nu|X_\nu-\varrho_*\ran$. Then we have $\nu\ge 0$, and
$\lan X-X_\nu|X_\lambda\ran=\lan X-X_\nu|X_\lambda-\varrho_*\ran=0$ for every $\lambda$. We see that
$$
\lan X|X_\lambda\ran
=\lan X_\nu| X_\lambda\ran =\tfrac 1{mn}+\lambda\nu(\lan\varrho|\varrho\ran-\tfrac 1{mn})
$$
is  nonnegative for every $\lambda\ge 0$. Since $C\subset \mathcal H^0\sqcup\mathcal H^+$, we see that
$C^\circ$ contains $X_\lambda$ for every $\lambda\ge 0$, and we conclude that $C^\circ$ is not compact.

For the \lq only if\rq\ part of (ii), we first note the identity
$$
\lan X|\varrho\ran=\lan X-\varrho_*|\varrho-\varrho_*\ran+\tfrac 1{mn},\qquad X,\varrho\in\mathcal H.
$$
From this, we see that the relation
$B(\varrho_*,a)\subset B(\varrho_*,b)^\circ$ holds whenever $ab\le \tfrac 1{mn}$.
Therefore $C\subset B(\varrho_ *,M)$ implies
$B(\varrho_*,\tfrac 1{mnM})\subset B(\varrho_*,M)^\circ\subset C^\circ$.
For the \lq if\rq\ part of (ii), we suppose that $\varrho_*$ is an interior point of $C^\circ$. Then $C^{\circ\circ}$
is compact by (i), which implies that $C$ is also compact, because $C$ is a closed subset of $C^{\circ\circ}$.

Now, it remains to prove (iii). We denote by
$$
\tilde C := \{ \lambda \varrho : \varrho \in C, \lambda \ge 0 \}
$$
the convex cone generated by $C$ and the origin. Then $\tilde C$ is a closed convex cone, since $C$ is compact.
We also denote by $(\tilde C)^\bullet$ the usual dual cone of $\tilde C$  temporarily:
$$
(\tilde C)^\bullet := \{Y \in (M_m \otimes M_n)_h :  \lan Y | X \ran \ge 0\ {\text{\rm for every}}\  X \in \tilde C \}.
$$
Then we have $C=\tilde C\cap\mathcal H=(\tilde C)^{\bullet\bullet}\cap\mathcal H$ by the usual duality
for the closed convex cones. We also have $(\tilde C)^\bullet \cap \mathcal H = C^\circ$, which implies
$C^{\circ\circ}=(\widetilde{C^\circ})^\bullet\cap\mathcal H$. Therefore, it suffices to show
$(\tilde C)^{\bullet\bullet}=(\widetilde{C^\circ})^\bullet$, which is equivalent to
\begin{equation}\label{xxxxx}
(\tilde C)^{\bullet}=\widetilde{C^\circ}.
\end{equation}
Since $\varrho_*$ is an interior point of $C$, it is also an interior point of $\tilde C$, and so we have
$\tr(Y)=mn\lan Y|\varrho_*\ran >0$ for every nonzero $Y\in(\tilde C)^{\bullet}$ by \cite[Proposition 2.3.1]{kye_lec_note}.
Using this, it is easy to see the relation (\ref{xxxxx}).
\end{proof}

It is well known \cite{Gurvits_Barnum} that $\varrho_*$ is an interior point of the convex set $\mathcal S_1$,
and so we see that $\blockpos_1=\mathcal S_1^\circ$ is compact. Therefore, the convex sets $\mathcal S_k$ and $\blockpos_k$
are compact, and $\varrho_*$ is an interior point of them for every $k=1,2,\dots, m\meet n$, where $m\meet n=\min\{m,n\}$.

Throughout this note, we consider compact convex sets in which the maximally mixed state $\varrho_*$
is an interior point with respect to the topology on $\mathcal H$.
For such a compact convex set $C$ in $\mathcal H$,
we define two numbers $\gamma^+[C]>0$ and $\gamma^-[C]<0$ satisfying
$$
X_\lambda\in C\ \Longleftrightarrow\ \gamma^-[C]\le\lambda\le\gamma^+[C].
$$
In order to consider the hyperplane $\mathcal H_\nu^0$ through $X_\nu$ which is perpendicular
to the one parameter family $\{X_\lambda\}$, we note that the following are equivalent;
\begin{itemize}
\item
$X\in\mathcal H^0_\nu$, that is, $\lan X - X_{\nu} | \varrho-\varrho_* \ran =0$, 
\item
$\lan X - X_{\nu} | \varrho \ran =0$, 
\item
$\lan X - X_{\nu} | X_\lambda \ran =0$ for every $\lambda$,
\item
$\lan X - X_{\nu} | X_\lambda \ran =0$ for some $\lambda \ne 0$.
\end{itemize}
Then $\mathcal H_\nu^0$ is a level set of the function $f_\lambda$
for every nonzero real numbers $\lambda$.

\begin{proposition}\label{pseudo_w}
Suppose that $\{X_\lambda\}$ is given by {\rm (\ref{one-para})}. For
a positive number $\nu>0$ and a compact convex set $C$
with an interior point $\varrho_*$, the following are equivalent;
\begin{enumerate}
\item[{\rm (i)}]
$\lan X_\nu | X_\lambda \ran\ge 0$ for every $\lambda\in [\gamma^-[C],\gamma^+[C]]$,
\item[{\rm (ii)}]
$\lan X_\nu | X_{\gamma^-[C]} \ran\ge 0$,
\item[{\rm (iii)}]
$C^\circ\cap \mathcal H_\nu^0$ is nonempty.
\end{enumerate}
For a negative number $\nu<0$, the above {\rm (i)}, {\rm (iii) and the following are equivalent:
\begin{enumerate}
\item[{\rm (ii)}${}^\prime$]
$\lan X_\nu | X_{\gamma^+[C]} \ran\ge 0$.
\end{enumerate}
}
\end{proposition}

\begin{proof}
For the brevity, we write $\gamma^+=\gamma^+[C]$ and $\gamma^-=\gamma^-[C]$.
The equivalence of (i) and (ii) is clear, since $\lambda\mapsto \lan X_\nu | X_\lambda \ran$ is increasing by $\nu > 0$.
Suppose that $W\in C^\circ\cap \mathcal H_\nu^0$. Then we have $\lan X_\nu | X_{\gamma^-} \ran=\lan W | X_{\gamma^-} \ran$ because $X_\nu$
and $W\in \mathcal H_\nu^0$ belong to the same level set of the function $X\mapsto \lan X | X_{\gamma^-} \ran$. This proves
(iii) $\Longrightarrow$ (ii), since we have $\lan W | X_{\gamma^-} \ran\ge 0$ by $X_{\gamma^-}\in C$ and
$W\in C^\circ$. For the direction (ii) $\Longrightarrow$ (iii),
We first note that the set
$$
\{ \lan W|X_{\gamma^-}\ran\in\mathbb R: W\in C^\circ\}
$$
is an interval, which is contained in $[0,+\infty)$, and contains $\lan \varrho_* | X_{\gamma^-}\ran=\frac 1{mn}$. This set also contains $0$,
because $X_{\gamma^-}$ is a boundary point of $C$. See \cite[Proposition 2.3.1]{kye_lec_note}.
Because $0\le \lan X_\nu | X_{\gamma^-} \ran \le \lan X_0 | X_{\gamma^-} \ran= \frac 1{mn}$ by $\nu > 0$ and $\gamma^-<0$,
we can take $W\in C^\circ$
such that $\lan X_\nu | X_{\gamma^-} \ran=\lan W | X_{\gamma^-} \ran$. Then $W\in \mathcal H_\nu^0$, and this completes the proof.
\end{proof}

For a given real number $\nu$, we define the open half-spaces
$$
\mathcal H^+_\nu=\{X\in\mathcal H: \lan X-X_\nu|\varrho-\varrho_*\ran > 0\},\qquad
\mathcal H^-_\nu=\{X\in\mathcal H: \lan X-X_\nu|\varrho-\varrho_*\ran < 0\}.
$$
Then $\mathcal H$ is the disjoint union $\mathcal H^-_\nu\sqcup \mathcal H^0_\nu \sqcup \mathcal H^+_\nu$.
If $\lambda>0$ then we have
$$
\mathcal H^+_\nu=\{X\in\mathcal H: \lan X-X_\nu|X_\lambda\ran > 0\},\qquad
\mathcal H^-_\nu=\{X\in\mathcal H: \lan X-X_\nu|X_\lambda\ran < 0\}.
$$
On the other hand, we have
$$
\mathcal H^+_\nu=\{X\in\mathcal H: \lan X-X_\nu|X_\lambda\ran < 0\},\qquad
\mathcal H^-_\nu=\{X\in\mathcal H: \lan X-X_\nu|X_\lambda\ran > 0\},
$$
whenever $\lambda<0$.

We denote by $\tilde\gamma^+[C^\circ]$ and $\tilde\gamma^-[C^\circ]$ the maximum and the minimum
of $\nu$'s satisfying conditions in Proposition \ref{pseudo_w}, respectively.
The inequalities $\tilde\gamma^-[C^\circ]\le \gamma^-[C^\circ]$ and $\gamma^+[C^\circ]\le \tilde\gamma^+[C^\circ]$
are clear by the property (iii).
Property (ii) of Proposition \ref{pseudo_w} tells us that $\nu=\tilde\gamma^+[C^\circ]$ if and only if
$\lan X_\nu | X_{\gamma^-[C]} \ran=0$. On the other hand, we see that
$C^\circ\subset \mathcal H^0_{\tilde\gamma^+[C^\circ]} \sqcup \mathcal H^-_{\tilde\gamma^+[C^\circ]} $ by the property (iii), and
$\mathcal H^0_{\tilde\gamma^+[C^\circ]} $ is a supporting hyperplane to $C^\circ$.
The hyperplane $\mathcal H^0_{\tilde\gamma^-[C^\circ]} $
also supports $C^\circ$.

\begin{theorem}\label{main}
Suppose that $C$ is a compact convex set in $\mathcal H$ with an interior point $\varrho_*$,
and two numbers $\nu>0$ and $\mu<0$ satisfy $\lan X_\nu | X_\mu\ran=0$. Then we have the following;
\begin{enumerate}
\item[{\rm (i)}]
$\nu=\tilde\gamma^+[C^\circ]$ if and only if $\mu=\gamma^-[C]$.
\item[{\rm (ii)}]
If $\nu$ satisfies the conditions in
Proposition \ref{pseudo_w} and $X_\mu\in C$, then we have $\nu=\tilde\gamma^+[C^\circ]$.
\end{enumerate}
For $\nu<0$ and $\mu>0$ satisfying $\lan X_\nu | X_\mu\ran=0$, we also have the following;
\begin{enumerate}
\item[{\rm (iii)}]
$\nu=\tilde\gamma^-[C^\circ]$ if and only if $\mu=\gamma^+[C]$.
\item[{\rm (iv)}]
If $\nu$ satisfies the conditions in
Proposition \ref{pseudo_w} and $X_\mu\in C$, then we have $\nu=\tilde\gamma^-[C^\circ]$.
\end{enumerate}
\end{theorem}

\begin{proof}
By the definition, we have $\lan X_{\tilde\gamma^+[C^\circ]} | X_{\gamma^-[C]} \ran = 0$.
Since the affine map $\lambda \mapsto \lan X_{\tilde\gamma^+[C^\circ]} | X_\lambda \ran$ is injective,
$\nu=\tilde\gamma^+[C^\circ]$ implies $\mu=\gamma^-[C]$.
The reverse direction is just the definition of $\tilde\gamma^+[C^\circ]$.

For the statement (ii), we note that $X_\mu\in C$ if and only if $W\in C^\circ$ implies $\lan W|X_\mu\ran\ge 0$
if and only if $W\in C^\circ$ implies $\lan W-X_\nu|X_\mu\ran\ge 0$
if and only if $C^\circ\subset \mathcal H^0_\nu\sqcup\mathcal H^-_\nu$.
The exactly same arguments work for the statements (iii) and (iv).
\end{proof}

We have studied the number  $\beta^\pm_k=\gamma^\pm[\blockpos_k]$ extensively in \cite{han_kye_beta}.
In this note, we define
$$
\tilde\beta^\pm_k:=\tilde\gamma^\pm[\blockpos_k],\qquad
\sigma^\pm_k:=\gamma^\pm[\mathcal S_k],\qquad
\tilde\sigma^\pm_k:=\tilde\gamma^\pm[\mathcal S_k],
$$
so that $X_\lambda\in\mathcal S_k$ if and only if $\sigma^-_k\le \lambda\le \sigma^+_k$.
holds.
Then we have
$$
\tilde\beta^-_k\le \beta^-_k\le \sigma^-_k<0
<\sigma^+_k\le \beta^+_k\le\tilde\beta^+_k.
$$
See {\sc Figure 2}.

\begin{figure}
\begin{center}
\setlength{\unitlength}{0.8 truecm}
\begin{picture}(12,5)
\drawline(-1,2)(12,2)
\put(4.88,2.1){$0$}
\put(5,2){\circle*{0.1}}

\put(3,2){\circle*{0.1}}
\put(6.2,2){\circle*{0.1}}
\qbezier(3,2)(3,4)(5,4)
\qbezier(5,4)(8.2,3.8)(6.2,2)
\qbezier(3,2)(3,0.1)(6.2,2)
\put(5.9,1.55){$\sigma^+_k$}
\put(3.12,1.55){$\sigma^-_k$}
\put(3.12,2.2){$\tilde\sigma^-_k$}
\put(7,2.2){$\tilde\sigma^+_k$}
\put(6.3,2.8){$\mathcal S_k$}
\drawline(3,-0.5)(3,5)
\drawline(6.98,-0.5)(6.98,5)
\put(6.98,2){\circle*{0.1}}

\put(7.4,1){$\mathcal D$}
\qbezier(5,0)(2,0)(2,2)
\qbezier(2,2)(2,4)(5,4)
\qbezier(5,4)(8,4)(8,2)
\qbezier(8,2)(8,0)(5,0)

\put(8.6,3.3){$\blockpos_{k}$}
\qbezier(5,5)(3.5,5)(1,2)
\qbezier(1,2)(0,0.5)(0.25,0.25)
\qbezier(0.25,0.25)(0.5,0)(5,0)
\qbezier(5,5)(10,5)(10,2)
\qbezier(10,2)(10,0)(5,0)

\put(1,2){\circle*{0.1}}
\put(10,2){\circle*{0.1}}

\drawline(10,-0.5)(10,5)
\put(10.1,1.45){$\beta^+_{k}$}
\put(10.1,2.2){$\tilde\beta^+_{k}$}

\drawline(0.2,-0.5)(0.2,5)
\put(0.2,2){\circle*{0.1}}
\put(-0.39,2.2){$\tilde\beta^-_{k}$}
\put(0.9,1.55){$\beta^-_{k}$}

\put(-0.4,0.1){$W$}
\put(0.21,0.39){\circle*{0.1}}

\end{picture}
\end{center}
\caption{The horizontal line represents the one parameter family $\{X_\lambda\}$, and the vertical lines represent
supporting hyperplane to the convex sets $\blockpos_k$ and $\mathcal S_k$.}
\end{figure}
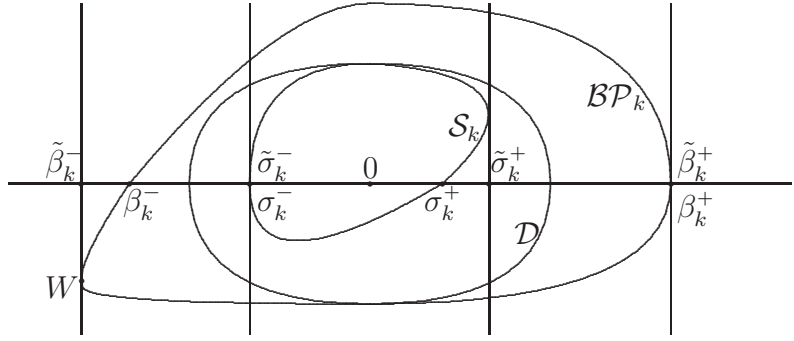

The facial structures of  the convex set $\mathcal D$ of all states is well known. We write $\delta^\pm=\gamma^\pm[\mathcal D]$,
and $\tilde\delta^\pm=\tilde\gamma^\pm[\mathcal D]$.
For a nontrivial subspace $E\subset \mathbb C^m\ot\mathbb C^n$,
the set $F_E$ of all states whose ranges are in $E$ is a face of $\mathcal D$ which is exposed, and every proper face of $\mathcal D$
arises in this way. We take the projection state $\varrho_E:=\frac 1{d}P_E$, where $P_E$ is the projection onto the subspace $E$
and  $d$ is the dimension of $E$. Then $\varrho_E$ is an interior point of $F_E$.
We show that the hyperplane through $\varrho_E$ which is perpendicular to the line connecting $\varrho_E$ and $\varrho_*$
is a supporting hyperplane to the convex set $\mathcal D$.
To see this, we take the one parameter family arising from $\varrho_E$. Then it is easy to see
$$
\delta^-=\beta^-_{m\meet n}=\frac {-d}{mn-d}.
$$
This also can be seen by the formulae (2) and (3) of \cite{han_kye_beta}, and we have $X_{\delta^-}=\varrho_{E^\perp}$,
the projection state onto the orthogonal complement $E^\perp$.
See  {\sc Figure 1} of \cite{han_kye_beta}. We also have $\delta^+=\sigma^+_{m\meet n}=1$, and
$$
\lan X_{\delta^-}|X_{\delta^+}\ran=\lan \varrho_{E^\perp} | \varrho_E\ran=0.
$$
Because $\mathcal D^\circ=\mathcal D$, we have
$$
\tilde\delta^-=\delta^-=\frac {-d}{mn-d},\qquad
\tilde\delta^+=\delta^+=1,
$$
by Theorem \ref{main} (i), (iii).
Therefore, we see that both hyperplanes $\mathcal H^0_{\delta^+}$ and $\mathcal H^0_{\delta^-}$
support the convex set $\mathcal D$.
We also note that a state $\varrho$ belongs to the face $F_E$ if and only if
$\lan \varrho-\varrho_E|\varrho_E\ran=0$, and so we conclude that $F_E=\mathcal D\cap\mathcal H^0_{\delta^+}$.

The number $\tilde\sigma^-_1$ may be even less than the number $\beta^-_1$ in general. Consider the diagonal state
$\varrho=p|00\ran\lan 00|+q|01\ran\lan 01|$ in the two qubit system with $p>q>0$ and $p+q=1$.
The linear map with $X_\lambda^\varrho$ as its Choi matrix has a commutative range, thus it is completely positive if and only if it is positive.
This implies that $||\varrho||_{S(1)} = ||\varrho|| = p$, thus
we have
$$
\beta^-_1=\frac {-1}{4p-1}.
$$
We consider the state $|11\ran\lan 11|\in \mathcal S_1$ then
$\lan |11\ran\lan 11|-X_\nu|\varrho\ran=0$ implies $\nu=\frac{-1}{8p^2-8p+3}$.
Since $\lan X_\nu | X_1 \ran = \lan X_\nu | \varrho \ran = \lan |11\ran\lan11| | \varrho \ran = 0$
and $X_1 \in \mathcal S_1$, we may apply
Theorem \ref{main} (iv) to see
$$
\tilde\sigma^-_1= \frac{-1}{8p^2-8p+3}.
$$
Therefore, we see that $\tilde\sigma^-_1<\beta^-_1$
by $\tfrac 12<p<1$.

\section{Variations of Werner states and isotropic states}

In this section, we consider the case when $k=1$ and $\varrho$ is a rank one projection, that is, a pure state.
By Schmidt decomposition, it is enough to
take $|\xi\ran=\sum_{i=0}^{n-1}p_i |ii\ran \in\mathbb C^n\ot\mathbb C^n$
with $\sum_{i=0}^{n-1}p_i^2=1$ and $p_0\ge p_1\ge\dots\ge p_{n-1}\ge 0$,
and put
$$
\varrho=|\xi\ran\lan\xi|
=\sum_{i=0}^{n-1}p_i^2|ii\ran\lan ii| +\sum_{i\neq j} p_i p_j |ii\ran\lan jj|
\in M_n\ot M_n.
$$
We recall that $\varrho$ is the Choi matrix of a completely positive map
of the form $\ad_s:a\mapsto s^*as$ with $s = \sum_{i=0}^{n-1} p_i |i\ran\lan i|$,
which is known \cite{{marcin_exp},{kye_exp}} to generate an exposed extreme ray of the convex cone of all positive maps.

\begin{theorem}\label{main_th}
Suppose that $\varrho=|\xi\ran\lan\xi|\in M_n\ot M_n$ is a pure state, and the Schmidt coefficients of $|\xi\ran$ is given by
$p_0\ge\cdots \ge p_{n-1}\ge 0$.
Then we have the following;
$$
\begin{aligned}
&\tilde\beta^-_1=-\frac{n^2p_0p_1+1}{n^2-1},\qquad
\beta^-_1=\dfrac{-1}{n^2p_0^2-1},\qquad
\delta^-=\sigma^-_1=\tilde\sigma^-_1=\dfrac{-1}{n^2-1},\\
& \sigma^+_1=\frac 1{n^2p_0p_1+1},\qquad
\tilde\sigma^+_1=\frac{n^2p_0^2-1}{n^2-1},\qquad
\delta^+=\beta^+_1=\tilde\beta^+_1=1.
\end{aligned}
$$
\end{theorem}

\begin{proof}
We have
$$
X_\lambda
=\sum_{i=0}^{n-1}\left(\dfrac{1-\lambda}{n^2}+\lambda p_i^2\right) |ii\ran\lan ii|
+\sum_{i\neq j}\left(\dfrac{1-\lambda}{n^2}\right) |ij\ran\lan ij|
+\sum_{i\neq j}\lambda p_i p_j |ii\ran\lan jj|.
$$
We also have $\|\varrho\|_{S(1)}=p_0^2$, and $\beta^-_1=\dfrac{-1}{n^2 p_0^2-1}$
by (8) and (3) in  \cite{han_kye_beta}.
Note that
$$
\begin{aligned}
X_{\beta^-_1}
&=\left(1-\dfrac{-1}{n^2 p_0^2-1}\right)\cdot\dfrac 1{n^2} I_{n^2}+ \dfrac{-1}{n^2 p_0^2-1}\varrho\\
&=\dfrac{p_0^2}{n^2p_0^2-1}\, I_{n^2} -  \dfrac{1}{n^2p_0^2-1}\varrho\\
&=\frac 1{n^2p_0^2-1}\left(p_0^2 I_{n^2}-\varrho\right).
\end{aligned}
$$

Now, we take $|\xi_{ij}\ran=\sqrt{p_i p_j}|ij\ran-\sqrt{p_i p_j}|ji\ran\in\mathbb C^n\ot\mathbb C^n$ for $i>j$, and put
$$
\varrho_{ij}=|\xi_{ij}\ran\lan\xi_{ij}|
=p_ip_j\left(|ij\ran\lan ij| +|ji\ran\lan ji| - |ij\ran\lan ji|- |ji\ran\lan ij|\right).
$$
We take the partial transpose $\varrho_{ij}^\Gamma$ of $\varrho_{ij}$ to get
$$
\varrho_{ij}^\Gamma
=p_ip_j \left(|ij\ran\lan ij| +|ji\ran\lan ji| - |ii\ran\lan jj|- |jj\ran\lan ii|\right),
$$
and we see that
$$
(n^2p_0^2-1)X_{\beta^-_1} = \sum_{i>j} \varrho_{ij}^\Gamma +D
$$
with a diagonal matrix $D$ with nonnegative entries.
Therefore,
$$
\frac 1{2p_ip_j}\varrho_{ij}^\Gamma=\frac 12(|ij\ran\lan ij|+|ji\ran\lan ji|-|ii\ran\lan jj| -|jj\ran\lan ii|)
$$
is a natural candidate of $1$-blockpositive matrix
which may give rise to the number $\tilde\beta^-_1$.
The next step is to look for $\nu$ so that $\lan \tfrac 1{2p_ip_j}\varrho_{ij}^\Gamma-X_\nu|\varrho\ran=0$. We have
$\lan \tfrac 1{2p_ip_j}\varrho_{ij}^\Gamma|\varrho\ran =-p_ip_j$
and
$$
\lan X_\nu|\varrho\ran
=\lan (1-\nu)\varrho_*+\nu\varrho|\varrho\ran=(1-\nu)\tfrac 1{n^2}+\nu=\tfrac 1{n^2}[1+(n^2-1)\nu].
$$
Therefore, we have $\nu=-\tfrac{n^2p_ip_j+1}{n^2-1}$. We take the lowest number
$$
\nu:=-\frac{n^2p_0p_1+1}{n^2-1}
$$
among them. We also take
$$
\mu:=\frac 1{1+n^2p_0p_1},
$$
which satisfies the relation $\lan X_\nu|X_{\mu}\ran=0$. We note that $\mu$ is the maximum of $\lambda$'s such that
$X_\lambda$ is of PPT. We have
\begin{equation}\label{X_mu}
X_\mu=\frac 1{1+n^2p_0p_1}
\left(\sum_{i=0}^{n-1} (p_0p_1+p_i^2)|ii\ran\lan ii|
+\sum_{i\neq j}p_0p_1 |ij\ran\lan ij|
+\sum_{i\neq k}p_i p_k|ii\ran\lan kk\ran\right).
\end{equation}

We proceed to show that $X_\mu$ is separable, in order to conclude that $\sigma^+_1=\mu$ and $\tilde\beta^-_1=\nu$
by Theorem \ref{main} (iii).
For a given $n$-tuple $\alpha=(\alpha_0,\alpha_1,\cdots,\alpha_{n-1})$ of complex numbers of modulus one, we take
$$
|\xi_\alpha\ran=\sum_{i=0}^{n-1} \sqrt{p_i}\alpha_i|i\ran\in\mathbb C^n,
$$
and
$$
|\eta_\alpha\ran=
|\xi_\alpha\ran|\bar \xi_\alpha\ran
=\sum_{i,j=0}^{n-1} \sqrt{p_i p_j}\alpha_i\bar \alpha_j|ij\ran\in\mathbb C^n\ot\mathbb C^n.
$$
Then we have
$$
\begin{aligned}
|\eta_\alpha\ran\lan \eta_\alpha|
&=\sum_{i,j,k,\ell=0}^{n-1} \sqrt{p_i p_j p_k p_\ell}\alpha_i\bar \alpha_j\bar\alpha_k\alpha_\ell |ij\ran\lan k\ell|\\
&=\sum_{i=0}^{n-1}p_i^2|ii\ran\lan ii|
  +\sum_{i\neq j}p_ip_j |ij\ran\lan ij|
  +\sum_{i\neq k}p_i p_k |ii\ran\lan kk|
  +\sum_{\text{\rm others}} A_{i,j,k,\ell}|ij\ran\lan k\ell|.
\end{aligned}
$$
We take $\alpha_t=\pm 1,\pm{\rm i}$ for each $t=0,1,\dots,n-1$, and averaging all of them to get
$$
\frac 1{4^n}\sum_\alpha |\eta_\alpha\ran\lan \eta_\alpha|
=\sum_{i=0}^{n-1}p_i^2|ii\ran\lan ii|
  +\sum_{i\neq j}p_ip_j |ij\ran\lan ij|
  +\sum_{i\neq k}p_i p_k |ii\ran\lan kk|.
$$
Comparing with (\ref{X_mu}), we have
$$
X_\mu=\frac 1{4^n(1+n^2p_0p_1)}\sum_\alpha |\eta_\alpha\ran\lan \eta_\alpha| +D,
$$
where $D$ is a diagonal matrix with nonnegative entries, since $\{p_i:i=0,1,\dots,n-1\}$ is decreasing.
This proves that $X_\mu$ is separable, and we conclude
$$
\sigma^+_1=\frac 1{1+n^2p_0p_1},\qquad
\tilde\beta^-_1=-\frac{n^2p_0p_1+1}{n^2-1}.
$$
Therefore, we see that $\mathcal H^0_{\tilde\beta^-_1}$ is a supporting hyperplane to the convex set $\blockpos_k$.
This is through
$$
\frac 1{2p_0p_1}\varrho^\Gamma_{10}
=\frac 12(|10\ran\lan 10|+|01\ran\lan 01|-|11\ran\lan 00|-|00\ran\lan 11|) \in\blockpos_1
$$
which is the Choi matrix of a completely copositive map
of the form $a\mapsto s^*a^\ttt s$ with $s = {1 \over \sqrt{2}} (|1\ran\lan0| - |0\ran\lan1|)$.

By the relation
$\lan X_{\tilde\sigma^+_1} | X_{\beta^-_1}\ran=0$, we have
$$
\tilde\sigma^+_1=\frac{n^2p_0^2-1}{n^2-1}.
$$
We also  note that $\lan |00\ran\lan 00|- X_{\tilde\sigma^+_1}|\varrho\ran=0$, and so
the hyperplane $\mathcal H^0_{\tilde\sigma^+_1}$ meets the convex set $\mathcal S_1$
at the separable state $|00\ran\lan 00|\in\mathcal S_1$.

Now, we proceed to determine $\sigma^-_1$ and $\tilde\beta^+_1$. For this purpose, it suffices to show that
$$
X_{\delta^-}=\tfrac 1{n^2-1}(I_{n^2}-\varrho)
$$
is separable. For $i=0,1,\dots,n-1$ with $i> j$ and a complex number $\alpha$ with $|\alpha|=1$, take
$$
\begin{aligned}
|\eta_{ij}^\alpha\ran
&=( \sqrt{p_j}|i\ran+\alpha \sqrt{p_i}|j\ran)\ot ( \sqrt{p_j}|i\ran-\bar \alpha \sqrt{p_i}|j\ran)\\
&= p_j|ii\ran -\bar\alpha  \sqrt{p_j p_i}|ij\ran
      +\alpha  \sqrt{p_i p_j}|ji\ran -  p_i|jj\ran
\in\mathbb C^n\ot\mathbb C^n,
\end{aligned}
$$
and the separable state
$$
\begin{aligned}
\varrho_{ij}
=\tfrac 14 \sum_{\alpha=\pm 1, \pm {\rm i}}|\eta_{ij}^\alpha\ran\lan \eta_{ij}^\alpha|
&=p_j^2|ii\ran\lan ii| +p_ip_j|ij\ran\lan ij| + p_ip_j|ji\ran\lan ji|+p_i^2|jj\ran\lan jj|\\
&\qquad -\left(p_i p_j |ii\ran\lan jj| +  p_ip_j |jj\ran\lan ii|\right).
\end{aligned}
$$
Summing up all of them, we get
$$
\begin{aligned}
\tilde\varrho
&=\sum_{i=0}^{n-1}(1-p_i^2)|ii\ran\lan ii| +\sum_{i\neq j} p_ip_j |ij\ran\lan ij|
-\left( |\xi\ran\lan\xi|-\sum_{i=0}^{n-1}p_i^2 |ii\ran\lan ii|\right)\\
&= \sum_{i=0}^{n-1}|ii\ran\lan ii| +\sum_{i\neq j}p_ip_j|ij\ran\lan ij| -|\xi\ran\lan\xi|,
\end{aligned}
$$
which is separable. Comparing with
$$
(n^2-1)X_{\delta^-}=I_{n^2}-\varrho=\sum_{i=0}^{n-1}|ii\ran\lan ii| +\sum_{i\neq j}|ij\ran\lan ij| -|\xi\ran\lan\xi|,
$$
we see that $(n^2-1)X_{\delta^-}$ is the sum of $\tilde\varrho$ and a diagonal matrix with nonnegative entries, and
we conclude that $X_{\delta^-}$ is separable. Therefore, we have
$$
\sigma^-_k=\delta^-=-\tfrac 1{n^2-1},\qquad \tilde\beta^+_k=\beta^+_k=\delta^+=1,
$$
for $k=1,2,\dots,n$, and see that $\mathcal H^0_1$ is a supporting hyperplane to $\blockpos_1$
through the state $X_1=\varrho$.
\end{proof}

In the inequalities $\tilde\beta^-_1\le\beta^-_1$ and $\sigma^+_1\le \tilde\sigma^+_1$
in Theorem \ref{main_th}, the equalities hold when
$(p_0^2,p_1^2)=(1,0)$ or $(\frac 1n,\frac 1n)$.
When $(p_0^2,p_1^2)=(\frac 1n,\frac 1n)$, we have isotropic states with
$$
\tilde\beta^-_1=\beta^-_1=\dfrac{-1}{n-1},\quad
\delta^-=\sigma^-_1=\dfrac{-1}{n^2-1},\quad
\sigma^+_1=\frac 1{n+1},\quad
\delta^+=\beta^+_1=\tilde\beta^+_1=1.
$$
See \cite[Section 1.7]{kye_lec_note}. When $(p_0^2,p_1^2)=(1,0)$,
we have
$$
\tilde\beta^-_1=\beta^-_1=\delta^-=\sigma^-_1=\frac{-1}{n^2-1},\qquad
\sigma^+_1=\delta^+=\beta^+_1=\tilde\beta^+_1=1.
$$
Since $\varrho$ is an extreme point of the convex set $\mathcal S_1$, we see that its dual face $\mathcal H^0_{\tilde\beta^-_1}\cap\blockpos_1$
is a maximal face of $\blockpos_1$, and every maximal face of $\blockpos_1$ arises in this way. See \cite{{kye-canad},{kye-korean}}.

We take the partial transpose $X_\lambda^\Gamma$ of $X_\lambda$, then we also have the following;
\begin{itemize}
\item
$X_\lambda^\Gamma$ is a state if and only if $-\tfrac 1{n^2p_0^2-1}\le\lambda\le\tfrac 1{n^2p_0p_1+1}$,
\item
$X_\lambda^\Gamma$ is separable if and only if $-\tfrac 1{n^2-1}\le\lambda\le\tfrac 1{n^2p_0p_1+1}$ if and only if $X_\lambda^\Gamma$
is of PPT.
\end{itemize}
When $p_i=\tfrac 1{\sqrt n}$ for $i=0,1,\dots, n-1$, we recover the Werner states.

\section{Conclusion}

In this note, we have considered the problem to find supporting hyperplanes to $k$-blockpositive matrices of trace one
whose perpendicular line is through the maximally mixed state, and showed that this problem is equivalent to
find the interval for states with Schmidt numbers not greater than $k$ on the line. When $k=m\meet n$,
we saw that supporting hyperplanes to density matrices are perpendicular to
the line between the two projection states arising from a subspace and its orthogonal complement.
In the general cases with $k< m\meet n$, it seems to be a challenging project to find
all supporting hyperplanes.

When $k=1$ and $\varrho$ is a pure state, we found supporting hyperplanes to $1$-blockpositive matrices of trace one
which is perpendicular to the one parameter family through the maximally mixed state $\varrho_*$ and $\varrho$.
We first determined the interval for $1$-blockpositivity and decomposed the blockpositive matrix at the endpoint
into the sum of partial transposes of product states.
When $\varrho$ is the maximally entangled state, this gives rise to the isotropic states, together with the Werner states
by taking the partial transpose. Our method gives a simple decomposition of separable Werner states
into the sum of product states.

\end{document}